\documentclass[12pt]{amsart}
\usepackage{amssymb,latexsym}
\usepackage{amsfonts}
\usepackage{amsmath}
\usepackage[colorlinks,linkcolor=blue,anchorcolor=blue,citecolor=blue]{hyperref}
\usepackage{algorithm}
\usepackage{enumerate}
\usepackage{algpseudocode}
\usepackage{verbatim}
\usepackage{graphicx}

\def\F{\mathbb{F}}

\def\Tr{\text{\rm Tr}}
\def\wt{\text{\rm wt}}
\def\Ker{\text{\rm Ker}}

\def\Im{\text{\rm Im}}

\def\x{\textrm{x}}
\def\y{\textrm{y}}
\def\u{\textrm{u}}
\def\D{\textrm{D}}
\def\c{\textrm{c}}
\def\Prj{\textrm{Prj}}

\newtheorem{theorem}{Theorem}[section]
\newtheorem{lemma}[theorem]{Lemma}
\newtheorem{example}[theorem]{Example}

\newtheorem{proposition}[theorem]{Proposition}

\theoremstyle{definition}

\newtheorem{remark}[theorem]{Remark}

\numberwithin{equation}{section}

\begin{document}

\title[Weight hierarchies]{Weight hierarchies of three-weight $p$-ary linear codes from inhomogeneous quadratic functions }

\author{Shupeng Hu}
\address{School of Mathematical Sciences, Qufu Normal University, Qufu Shandong, {\rm 273165}, China}

\author{Fei Li}
\address{Faculty of School of Statistics and Applied Mathematics,
Anhui University of Finance and Economics, Bengbu, {\rm 233041}, Anhui, P.R.China}
\email{cczxlf@163.com}

\author{Xiumei Li}
\address{School of Mathematical Sciences, Qufu Normal University, Qufu Shandong, {\rm 273165}, China}
\email{lxiumei2013@qfnu.edu.cn}



\begin{abstract}
The weight hierarchy of a linear code have been an important research topic in coding theory since Wei's original work in 1991.
In this paper, choosing $ D=\Big\{(x,y)\in \Big(\F_{p^{s_1}}\times\F_{p^{s_2}}\Big)\Big\backslash\{(0,0)\}: f(x)+\Tr_1^{s_2}(\alpha y)=0\Big\}$ as a defining set ,
where $\alpha\in\mathbb{F}_{p^{s_2}}^*$ and $f(x)$ is a quadratic form over $\mathbb{F}_{p^{s_1}}$ with values in $\F_p$, whether $f(x)$ is non-degenerate or not,
we construct a family of three-weight $p$-ary linear codes and
determine their weight distributions and weight hierarchies completely.
Most of the codes can be used in secret sharing schemes.
\end{abstract}

\maketitle

\section{Introduction}

\label{intro}


For an odd prime number $p$ and a positive integer $s$, let $ \mathbb{F}_{p^s} $ be the finite field with $ p^{s} $
elements and $\mathbb{F}_{p^s}^{*}$ be its multiplicative group.

An $[n,k,d]$ $p$-ary linear code $C$ is a $k$-dimensional subspace of $ \mathbb{F}_{p}^{n} $ with minimum (Hamming) distance $d$.
For $i\in\{1,2,\cdots,n\}$, denote by $A_{i}$ the number of codewords  in $C$ with Hamming weight $i$. The weight distribution of $C$ is defined by the sequence $(1,A_{1},\cdots,A_{n})$
and the weight enumerator of $C$ is defined by the polynomial $1+A_{1}x+A_{2}x^{2}+\cdots+A_{n}x^{n}$.
If the sequence $(A_{1},\cdots,A_{n})$ has $t$ nonzero $A_{i}$ with $ i=1,2,\cdots n$, then the code $C$ is called a $t$-weight code.
In coding theory, the weight distribution of linear codes is a classical research topic and attracts much attention. Furthermore,
linear codes with a few weights have important applications in authentication codes \cite{8DH07},
association schemes \cite{4CG84}, secret sharing \cite{21YD06} and strongly regular graphs \cite{5CK86}.


For an $[n,k,d]$ linear code $C$ and $1\leq r\leq k$, the concept of generalized Hamming weights (GHW) $d_r(C)$ can be viewed as an extension of Hamming weight.
We recall the definition of the generalized Hamming weights of linear codes.
Denote by
$ [C,r]_{p} $ the set of all its $\mathbb{F}_{p}$-vector subspaces with dimension $r$.
For $ H \in [C,r]_{p}$, the support $ \textrm{Supp}(H)$ of $H$ is the set of coordinates where not all codewords of $H$ are zero, that is,
$$ \textrm{Supp}(H)=\Big\{i:1\leq i\leq n, c_i\neq 0 \ \ \textrm{for some $c=(c_{1}, c_{2}, \cdots , c_{n})\in H$}\Big\}.$$
The $r$-th generalized Hamming weight of $C$, which is also called the $r$-th minimum support weight, is
$$
d_{r}(C)=\min\Big\{|\textrm{Supp}(H)|:H\in [C,r]_{p}\Big\}, \ 1\leq r\leq k.
$$
The set $ \{d_{1}(C),d_{2}(C),\cdots,d_{k}(C)\}$ is called the weight hierarchy of $C$.

The generalized Hamming weights was introduced in 1977 by Helleseth et al. \cite{9HK92,K78} and Victor Wei \cite{20WJ91} proved that they can characterize the cryptography performance of a linear code over the wiretap channel of type II and determined the weight hierarchies for some well-known classes of codes, such as Hamming codes, Reed-Muller
codes, Reed-Solomon codes and Golay codes etc. From then on, much more attention was paid to generalized Hamming weights (see \cite{3CC97,11HP98,14JL97,20WJ91,27WZ94}). To determine the weight hierarchy of linear codes is relatively challenging, during the past decade, there were research results
about weight hierarchies of some classes of linear codes \cite{2B19,13JF17,18LF17,19LF17,LL20-0,LL20,LL22,19LW19,21XL16,22YL15}.

The rest of this paper is organized as follows.
In Sec. 2, we introduce a generalized method of constructing linear code by defining sets and
give a corresponding formula for computing the generalized Hamming weights.
In Sec. 3, we construct $p$-ary linear codes with three weights and determine their weight distributions and weight hierarchies completely.
In Sec. 4, we summarize the paper.

\section{Preliminaries}
\subsection{The defining-set construction of linear codes}

Ding et al. \cite{6DD14} proposed a generic construction of linear codes as below.
Denote by $\Tr_1^s$ the trace function from $\mathbb{F}_{p^s}$ onto $\mathbb{F}_p$.
Let $ D= \{d_{1},d_{2},\cdots,d_{n}\}$ be a subset of $\mathbb{F}_{p^s}^{\ast}.$
A $p$-ary linear code of length $n$ is defined as follows:
\begin{eqnarray*}\label{defcode1}
         C_{D}=\{\left( \Tr_1^s(xd_1), \Tr_1^s(xd_2),\ldots, \Tr_1^s(xd_{n})\right):x\in \mathbb{F}_{p^{s}}\},
\end{eqnarray*}
and $D$ is called the defining set of $C_{D}$. By choosing some proper defining sets, some optimal linear
codes with a few weights can be constructed (see \cite{5DJ16,25DD15,9DL16,JL19,19TXF17,24ZL16} and the references therein).

Li et al. generalized Ding's constructing method and also constructed some linear codes with a few weights \cite{LY17}. The following constructing method can be regarded as a generalized form of Li's.
For some positive integers $s_1,s_2,\ldots,s_e$, define $s=\sum_{i=1}^e s_i$, $\mathbb{F}=\mathbb{F}_{p^{s_1}}\times\mathbb{F}_{p^{s_2}}\times\cdots\times\mathbb{F}_{p^{s_e}}$ and $\mathbb{F}^\star=\mathbb{F}\Big\backslash\{(0,0,\cdots,0)\}$.
For $X=(x_{1},x_{2},\cdots,x_{e}), Y=(y_{1},y_{2},\cdots,y_{e})\in \mathbb{F}$,
denote by $\Tr (X\odot Y)$ as follows,
$$
\Tr (X\odot Y)=\Tr_1 ^{s_1}(x_1y_1)+\Tr_1^{s_2}(x_2y_2)+\cdots+\Tr_1^{s_e}(x_e y_e).
$$

 Let $D=\{Y_1,Y_2,\cdots,Y_n\}$
be contained in $\mathbb{F}^\star$.
Define
a $p$-ary linear code $ C_{D} $ as follows:
\begin{eqnarray}\label{defcode1}
         C_{D}=\Big\{\Big(\Tr (X\odot Y_i\Big)_{Y_i\in D}:X
         \in \mathbb{F}\Big\}.
\end{eqnarray}
Here, $ D $ is also called the defining set. Now we give a formula to calculate the $r$-th generalized Hamming weight of
$ C_{\D}$ in \eqref{defcode1}, which
is a generalization of \cite[Proposition 1]{LL20}.

\begin{proposition}\label{pro:d_r}
For each $ r $ and $ 1\leq r \leq s$, if the dimension of $ C_{D} $
is $ s$, then
\begin{equation}\label{eq:d_r}
d_{r}(C_{D})= n-\max\big\{|D \cap H|: H \in [\mathbb{F},s-r]_{p}\big\}.
\end{equation}
\end{proposition}

\begin{proof}
Let $ \sigma$ be a map from $ \mathbb{F} $ to $ \mathbb{F}_{p}^{n}$ defined by
$$
\sigma(X)=\Big(\Tr (X\odot Y_1),\Tr (X\odot Y_2),\cdots,\Tr (X\odot Y_n)\Big).
$$
Obviously, $\sigma$ is an $\F_p$-linear homomorphism and the image $\sigma(\mathbb{F})$ is $C_{D}$. By the assumed condition, that is, the dimension of $ C_{D} $
is $s$, we know that $\sigma$ is injective. Taking any $r$-dimensional subspace $ U_{r} \in [C_{D}, r]_{p}$, denote by $H_{r} = \sigma^{-1}(U_{r})$
the inverse image of $U_{r}$, then $H_{r}$ is an $r$-dimensional
subspace of $\F$. Let $ \{ \beta_{1},\beta_{2},\ldots ,\beta_{r}\}$ be an $\mathbb{F}_{p}$-basis of $ H_{r}$, we have
$$
d_{r}(C_{\D})= n-\max\{N(U_{r}): U_{r} \in [C_{\D}, r]_{p}\},
$$
where
\begin{align*}
N(U_{r})&=\Big|\{i: 1\leq i \leq n, c_{i}=0\ \ \textrm{for any\ } C=(c_{1},c_{2}, \ldots, c_{n}) \in U_{r}\}\Big|  \\
&=\Big|\{i: 1\leq i \leq n, \Tr (X\odot Y_i)=0\ \ \textrm{for any\ } X \in H_{r}\}\Big| \\
&=\Big|\{i: 1\leq i \leq n, \Tr(\beta_j \odot Y_{i})=0, j =1,2,\cdots,r\}\Big|.
\end{align*}
Then, by the orthogonal property of additive characters, we have
\begin{align}\label{N(U)}
N(U_{r})
&=\frac{1}{p^{r}}\sum_{i=1}^{n}\sum_{x_{1}\in \mathbb{F}_{p}}\zeta_{p}^{x_{1}\Tr\big((\beta_{1}\odot Y_{i})\big)}\ldots \sum_{x_{r}\in \mathbb{F}_{p}}\zeta_{p}^{{x_{r}\Tr\big((\beta_{r}\odot Y_{i})\big)}} \nonumber \\
&=\frac{1}{p^{r}}\sum_{i=1}^{n}\sum_{x_{1},\ldots,x_{r}\in \mathbb{F}_{p}}\zeta_{p}^{\mathrm{Tr}\big( (x_1\beta_{1}+\ldots+x_{r}\beta_{r})\odot Y_{i}\big)}    \nonumber\\
&=\frac{1}{p^{r}}\sum_{i=1}^{n}\sum_{\beta\in H_{r}}\zeta_{p}^{\mathrm{Tr}(\beta\odot Y_{i})}.
\end{align}

For $H_r$,
let $
H_r^{\perp}=\{X\in \mathbb{F}: \mathrm{Tr}(X \odot Y)=0 \ \textrm{for any}\ Y \in H_r\}$,
by Theorem 2.37 \cite{16LN97}, one can prove that $ \dim_{\mathbb{F}_{p}}(H_r)+\dim_{\mathbb{F}_{p}}(H_r^{\bot})=s$ and $(H_r^{\bot})^{\bot}=H_r$. So,
for any $Y\in \mathbb{F}$,
\begin{equation}\label{H}
\sum_{\beta\in H_{r}}\zeta_{p}^{\Tr(\beta\odot Y)}=\left\{\begin{array}{ll}
|H_{r}|, & \textrm{if\ } \ Y \in H_{r}^{\bot}, \\
0, & \textrm{otherwise\ }.
\end{array}
\right.
\end{equation}
Combining \eqref{N(U)} with \eqref{H}, we have
$$
N(U_{r})
=\frac{1}{p^{r}}\sum_{Y\in D \cap H_{r}^{\bot}} |H_{r}|=|D \cap H_{r}^{\bot}|.
$$
Finally, we get the desired result, which follows from the fact that there is a bijection
between $[\mathbb{F},r]_{p}$ and $[\mathbb{F},s-r]_{p}$.

\end{proof}

\subsection{Some notation fixed throughout this paper}\label{subsec:notation}
For convenience, we fix the following notation. For basic results on cyclotomic field $ \mathbb{Q}(\zeta_{p}) $, one is referred to \cite{IR90}.
\begin{itemize}
\item Let $\mathrm{Tr}_1^s$ be the trace function from $\mathbb{F}_{q}$ to $\mathbb{F}_{p}$.
Namely, for each $x\in \mathbb{F}_{q}$,
$$
\mathrm{Tr}_1^s(x)=x+x^{p}+ \cdots +x^{p^{e-1}}.
$$
\item $p^{\ast}=(-1)^{\frac{p-1}{2}}p$.
\item $\zeta_{p}=\exp(\frac{2\pi i}{p})$ is a primitive $p$-th root of unity.
\item $\upsilon$ is the integer-valued function on $\F_p$ satisfying $\upsilon(0)=p-1$ and $\upsilon(z)=-1$ for $z\in \mathbb{F}_{p}^{\ast}$.
\item $\bar{\eta}$ is the quadratic character of $\mathbb{F}_{p}^{\ast}$.
It is extended by letting $\bar{\eta}(0)=0$.
\item Let $\mathbb{Z}$ be the rational integer ring and $\mathbb{Q}$ be the rational field. Let $\mathbb{K}$ be the cyclotomic
field $\mathbb{Q}(\zeta_{p})$. The field extension $\mathbb{K}/\mathbb{Q}$ is Galois of degree $p-1$.
The Galois group $\mathrm{Gal}(\mathbb{K}/\mathbb{Q})=\Big\{\sigma_{z}: z\in (\mathbb{Z}/ p\mathbb{Z})^{\ast}\Big\}$, where the automorphism
$\sigma_{z}$ is defined by $\sigma_{z}(\zeta_{p})=\zeta_{p}^{z}$.
\item $\sigma_{z}(\sqrt{p^{\ast}})=\bar{\eta}(z)\sqrt{p^{\ast}}$, for $1\leq z \leq p-1$.
\item Let $\Big\langle\alpha_{1},\alpha_{2},\cdots,\alpha_{r}\Big\rangle$ denote a space spanned by $\alpha_{1},\alpha_{2},\cdots,\alpha_{r}$.
\end{itemize}

\begin{lemma}[{\cite[Lemma 4]{19TXF17}}]\label{lem:7}
 With the symbols and notation above, for any $z\in \mathbb{F}_{p}$, we have the following.

$$
\sum\limits_{y\in \mathbb{F}_{p}^{\ast}}\sigma_{y}((p^{\ast})^{\frac{r}{2}}\zeta_{p}^{z})
=\left\{\begin{array}{ll}
\bar{\eta}(-z)p^r(p^{\ast})^{-\frac{r-1}{2}}, & \textrm{if\ } r  \ \textrm{is odd\ }, \\
\upsilon(z)p^r(p^{\ast})^{-\frac{r}{2}},  & \textrm{if\ } r \ \textrm{is even\ }.
\end{array}
\right.
$$

\end{lemma}

\subsection{Quadratic form}

Viewing $\mathbb{F}_{q}$ with $q=p^s$ as an $\mathbb{F}_{p}$-linear space and fixing $\upsilon_{1},\upsilon_{2},\cdots,\upsilon_{s} \in \mathbb{F}_{q} $ as its $\F_p$-basis, then
for any $x=x_{1}\upsilon_{1}+x_{2}\upsilon_{2}+\cdots+x_{s}\upsilon_{s} \in \mathbb{F}_{q}$ with $x_i\in\F_p, i=1,2,\cdots,s$,
there is an $\F_p$-linear isomorphism $\F_q\simeq\F_p^s$ defined as:
$$x=x_{1}\upsilon_{1}+x_{2}\upsilon_{2}+\cdots+x_{s}\upsilon_{s} \mapsto X=(x_1,x_2,\cdots,x_s),$$
where $X=(x_1,x_2,\cdots,x_s)$ is called the coordinate vector of $x$ under the basis $v_1,v_2,\cdots,v_s$ of $\F_q$.

Let $f$ be a quadratic form over $\mathbb{F}_{q}$ with values in $\mathbb{F}_{p}$ and \begin{equation*}\label{eq:F}
    F(x,y)=\frac{1}{2}\Big(f(x+y)-f(x)-f(y)\Big), \textrm{for any}\  x,y\in\mathbb{F}_q,
\end{equation*} then $f$ can be represented by
\begin{align}\label{eq:f}
f(x)=f(X)=f(x_{1},x_{2},\cdots,x_{s})&=\sum_{1\leq i,j\leq s}F(v_i,v_j)x_{i}x_{j}
=XAX^T,
\end{align}
where $A=(F(v_i,v_j))_{s\times s}, F(v_i,v_j)\in \mathbb{F}_{p}, F(v_i,v_j)=F(v_j,v_i)$ and $X^T$ is the transposition of $X$. Denote by $R_f=\textrm{Rank}\ A$ the rank of $f$, we say that $f$ is non-degenerate if $R_f=s$ and degenerate, otherwise.
For any $x,y \in\F_q$, $x,y$ can be uniquely expressed as $x=x_{1}\upsilon_{1}+x_{2}\upsilon_{2}+\cdots+x_{s}\upsilon_{s}$ and $y=y_{1}\upsilon_{1}+y_{2}\upsilon_{2}+\cdots+y_{s}\upsilon_{s}$  with $x_i,y_i\in\F_p$.
Hence, we have
\begin{align}\label{eq:F}
    F(x,y)&=\frac{1}{2}\Big(f(x+y)-f(x)-f(y)\Big) \nonumber\\
    &=\frac{1}{2}\Big(f(X+Y)-f(X)-f(Y)\Big)\nonumber\\
    &=XAY^T=YAX^T\nonumber\\
&=\Tr_1^s\Big(xL_{f}(y)=\Tr_1^s\Big(yL_{f}(x)\Big)\Big),
\end{align}
where $L_{f}$ is a linearized polynomial over $\mathbb{F}_q$
defined as
\begin{eqnarray}\label{eq:L_f}
        L_{f}(\upsilon_{1},\upsilon_{2},\cdots,\upsilon_{s})=(\upsilon_{1},\upsilon_{2},\cdots,\upsilon_{s})\Big(\Tr_1^s(v_iv_j)\Big)^{-1}A.
\end{eqnarray}
Let $\mathrm{Im}(L_{f})=\Big\{L_{f}(x):x\in\F_q\Big\},\ \Ker(L_{f})=\Big\{x\in\F_q:L_{f}(x)=0\Big\}$ denote the image and kernel of $L_{f}$, respectively.
If $b\in \mathrm{Im}(L_{f})$, we denote $x_{b}\in \mathbb{F}_q$ with satisfying $L_{f}(x_{b})=-\frac{b}{2}$.

From eq. \eqref{eq:F} and \eqref{eq:L_f}, we have
\begin{equation*}
\Ker(L_{f})=\{x\in \mathbb{F}_q:f(x+y)=f(x)+f(y), \textrm{for any}\  y\in\mathbb{F}_q\}
\end{equation*}
and $\textrm {rank}\ L_f = R_f$.

The following lemma will play an important role in settling the weight distributions and weight hierarchies.
For more details, one is referred to \cite{19TXF17}.

\begin{lemma}[{\cite[Lemma 5]{19TXF17}}]\label{lem:6}
 Let the symbols and notation be as above and $f$ be defined in \eqref{eq:f} and $b\in \mathbb{F}_{q}$. Then
$$
\sum\limits_{x\in \mathbb{F}_{q}}\zeta_{p}^{f(x)-\mathrm{Tr}_1^s(bx)}
=\left\{\begin{array}{ll}
0, & \textrm{if\ } b\notin \mathrm{Im}(L_{f}), \\
\varepsilon_{f}(p^{\ast})^{\frac{R_{f}}{2}}p^{s-R_f}\zeta_{p}^{-f(x_{b})},  & \textrm{if\ } b\in \mathrm{Im}(L_{f}).
\end{array}
\right.
$$
where $x_{b}$ satisfies $L_{f}(x_{b})=-\frac{b}{2}$.
\end{lemma}

Let $f$ be a quadratic form over $\mathbb{F}_{q}$ defined by eq.\eqref{eq:f}. We can find an invertible matrix $M$ over $\F_p$ such that
$$
MAM^T=\textrm{diag}(\lambda_{1},\lambda_{2},\cdots,\lambda_{R_{f}},0,\cdots,0)
$$
is a diagonal matrix, where $\lambda_{1},\lambda_{2},\cdots,\lambda_{R_{f}}\in\F_p^*$.
Let $\Delta_{f}=\lambda_{1}\lambda_{2}\cdots\lambda_{R_{f}}$, and $\Delta_{f}=1$ if $R_{f}=0.$
We call $\bar{\eta}(\Delta_{f})$ the sign $\varepsilon_{f}$ of the quadratic form $f$.
It is an invariant under nonsingular linear transformations in matrix.

For an $r$-dimensional subspace $H$ of $\F_q$, its dual space $H^{\perp_f}$ is defined by
$$
H^{\perp_f}=\Big\{x\in \mathbb{F}_{q}:\ F(x,y)=0, \ \mbox{for  any} \  y \in H\Big\}.
$$
It's well known that if $f$ is non-degenerate then
\begin{equation}\label{dim:H}\dim (H) +\dim (H^{\perp_f}) = s. \end{equation}
For $H=\F_q$, it is known that $R_f+\dim(\F_q^{\perp_f})=s$, whether $f$ is non-degenerate or not.
So $R_f$ can also be defined as the codimension of $\F_q^{\perp_f}$.
Restricting the quadratic form $f$ to $H$,
it becomes a quadratic form denoted by $f|_{H}$ over $H$ in $r$ variables. Similarly, we define the dual space $H_{f|_{H}}^{\perp_f}$ of $H$
under $f|_{H}$ in itself by
$$
H_{f|_{H}}^{\perp_f}=\Big\{x\in H:\ F(x,y)=0, \ \mbox{for  any} \  y \in H\Big\}.
$$
Let $R_{H}$ and $\varepsilon_{H}$ be the rank and sign of $f|_{H}$ over $H$, respectively. Obviously, $H_{f|_{H}}^{\perp_f}=H\bigcap H^{\perp_f}$ and
\begin{equation}\label{rk:H}R_H=r-\dim(H_{f|_{H}}^{\perp_f}). \end{equation}

 For $\beta\in \mathbb{F}_{p}$, set $D_{\beta}=\Big\{x\in \mathbb{F}_{q}|f(x)=\beta\Big\}$,
we shall give some lemmas, which are essential to prove our main
results.

\begin{lemma}[{\cite[Lemma 2]{LL20-0}}]\label{lem:1}
Let $f$ be a quadratic form over $\mathbb{F}_{q}, \beta\in \mathbb{F}_{p}$ and $H$ be an $r$-dimensional nonzero subspace of $\mathbb{F}_{q}$, then
$$
|H\cap D_{\beta}|=\left\{\begin{array}{ll}
p^{r-1}+\upsilon(\beta)\overline{\eta}((-1)^{\frac{R_{H}}{2}})\varepsilon_{H}p^{r-\frac{R_{H}+2}{2}},  &\textrm{if\ } \ R_{H}\equiv0\pmod 2, \\
p^{r-1}+\overline{\eta}((-1)^{\frac{R_{H}-1}{2}}\beta)\varepsilon_{H}p^{r-\frac{R_{H}+1}{2}},  &\textrm{if\ } \ R_{H}\equiv1\pmod 2,
\end{array}
\right.
$$
where $\upsilon(\beta)=p-1$ if $\beta=0$, otherwise $\upsilon(\beta)=-1$.
\end{lemma}

\begin{lemma}\label{lem:2}
Let $f$ be a quadratic form over $\mathbb{F}_{q}$ with the rank $R_f$. There exists an $e_0$-dimensional subspace
$H $ of $ \mathbb{F}_{q}$ such that $\F_q^{\perp_f}\subseteq H$ and $f(x)=0$ for any $x\in H$, where

$
e_0=\left\{\begin{array}{ll}
s-\frac{R_f+1}{2}, & \textrm{if $R_f$ is odd}, \\
s-\frac{R_f}{2}, & \textrm{if $R_f$ is even and $\varepsilon_{f}=(-1)^{\frac{R_f(p-1)}{4}}$},\\
s-\frac{R_f+2}{2}, & \textrm{if $R_f$ is even and $\varepsilon_{f}=-(-1)^{\frac{R_f(p-1)}{4}}$}.
\end{array}
\right.
$
\end{lemma}
\begin{proof} If $f$ is non-degenerate, the desired conclusion directly from Proposition 2 and Proposition 3 \cite{{19LF17}}.
If $f$ is degenerate, then the induced quadratic form $\overline{f}$ of $f$ is non-degenerate  over $\overline{\F}_{q}=\F_{q}\Big/\F_{q}^{\perp_f}$. Applying Proposition 2 and Proposition 3 \cite{{19LF17}} to $\overline{f}$, we can obtain a subspace $\overline{H}$ of $\overline{\F}_{q}$, which satisfies the condition of this lemma in $\overline{\F}_{q}$. Let $\varphi$ be the canonical map from $\F_{q}$ onto $\overline{\F}_{q}$, then the inverse image of $\overline{H}$ in $\F_{q}$ is the desired $H$.
\end{proof}

\begin{lemma}\label{lem:2-2}
Let $f$ be a quadratic form over $\mathbb{F}_{q}$ with the rank $R_f$. For each $\beta\in \F_p^*$, there exists an $l_0$-dimensional subspace
$H$ of $ \mathbb{F}_{q}$ such that $R_H=1, \varepsilon_{H}=\bar{\eta}(\beta)$ and $H \cap\ \F_q^{\perp_f}=\{0\}$, where $
l_0=\left\{\begin{array}{ll}
\frac{R_f-1}{2}, & \textrm{if $R_f$ is odd}, \\
\frac{R_f}{2}, & \textrm{if $R_f$ is even }.
\end{array}
\right.
$
\end{lemma}
\begin{proof} We only prove the case of $f$ degenerate. Let $\bar{f}$ be the induced quadratic form of $f$ in $\overline{\F}_{q_1}$, then $\bar{f}$ is non-degenerate over $\overline{\F}_{q}$ with rank $\varepsilon_{\bar{f}} = R_f$. Applying Proposition 2 \cite{{19LF17}} to $\bar{f}$, there is an $(l_0-1)$-dimensional subspace $\overline{H}_{l_0-1}\subset \overline{\F}_q$ such that $\overline{H}_{l_0-1}\subset \overline{H}_{l_0-1}^{\perp_{\bar{f}}}$. By \eqref{dim:H} and \eqref{rk:H},
we have $\dim(\overline{H}_{l_0-1}^{\perp_{\bar{f}}})=R_f-(l_0-1)$ and
\begin{align*}
   R_{\overline{H}_{l_0-1}^{\perp_{\bar{f}}}}&=\dim(\overline{H}_{l_0-1}^{\perp_{\bar{f}}})-\dim(\overline{H}_{\bar{f}|_{\overline{H}_{l_0-1}^{\perp_{\bar{f}}}}}^{\perp_{\bar{f}}})\\
   &=R_f-(l_0-1)-\dim(\overline{H}_{l_0-1}^{\perp_{\bar{f}}}\bigcap (\overline{H}_{l_0-1}^{\perp_{\bar{f}}})^{\perp_{\bar{f}}})\\
   &=R_f-2(l_0-1)\geq 2.
\end{align*}
Applying Lemma \ref{lem:1} to $\overline{H}_{l_0-1}^{\perp_{\bar{f}}}$, we have $|\overline{H}_{l_0-1}^{\perp_{\overline{f}}}\cap \overline{D}_{\beta}| >1$, where $\overline{D}_{\beta}=\Big\{\overline{x}\in \overline{\mathbb{F}}_{q}|\overline{f}(\bar{x})=\beta\Big\}$. We choose an element $\overline{\gamma}\in \overline{H}_{l_0-1}^{\perp_{\bar{f}}}\cap \overline{D}_{\beta}$ and define
$\overline{H}=<\overline{\gamma}>+\overline{H}_{l_0-1}$, then $\dim(\overline{H})=l_0, R_{\overline{H}}=1$ and $\varepsilon_{\overline{H}}=\bar{\eta}(\beta)$.
Taking $\overline{\gamma}_1,\overline{\gamma}_2,\cdots,\overline{\gamma}_{l_0}$ of $\overline{H}$ as its $\F_p$- basis, define $H=\Big\langle\gamma_1,\gamma_2,\cdots,\gamma_{l_0}\Big\rangle$, then the $H$ is our desired subspace.
\end{proof}

\section{Linear codes from inhomogeneous quadratic functions}

In this section, we fix the following notation. Let $s_i$ be positive integers, $q_i= p^{s_i}, i=1,2$. Denote $\mathbb{F}=\mathbb{F}_{q_1}\times\mathbb{F}_{q_2},\mathbb{F}^\star=\mathbb{F}\Big\backslash\{(0,0)\}$ and $s=s_1+s_2$. We study the weight distribution and weight hierarchy of $C_{D}$ in \eqref{defcode1}, where their defining set are
\begin{align}\label{set:D1}
D=\Big\{(x,y)\in \mathbb{F}^\star: f(x)+\Tr_1^{s_2}(\alpha y)=0\Big\},
\end{align}
where $\alpha\in\mathbb{F}_{q_2}^*$ and $f(x)$ is a quadratic form over $\mathbb{F}_{q_1}$ defined in \eqref{eq:f}.

\subsection{The weight distribution of the presented linear code}

In this subsection, we first calculate the length of $C_{D}$ defined in \eqref{defcode1} and the Hamming weight of non-zero codewords of $ C_{D}$.

\begin{lemma}\label{lem:length}
Let $D$ be defined as above and $C_{D}$ be defined in \eqref{defcode1}, define $n=|D|$. Then,
$$n=p^{s-1}-1.$$
\end{lemma}

\begin{proof}
By the orthogonal property of additive characters, we have
\begin{align*}
n
&=\frac{1}{p}\sum_{(x,y)\in \mathbb{F}}\sum_{u\in \mathbb{F}_{p}}\zeta_{p}^{u\big(f(x)+\Tr_1^{s_2}(\alpha y)\big)}-1  \\
&=\frac{1}{p}\sum_{(x,y)\in \mathbb{F}}\Big(1+\sum_{u\in \mathbb{F}_{p}^{\ast}}\zeta_{p}^{u\big(f(x)+\Tr_1^{s_2}(\alpha y)\big)}\Big)-1  \\
&=p^{s-1}+\frac{1}{p}\sum_{u\in \mathbb{F}_{p}^{\ast}}\sum_{y\in \mathbb{F}_{q_2}}\zeta_{p}^{\Tr_1^{s_2}(u\alpha y)}\sum_{x\in \mathbb{F}_{q_1}}\zeta_{p}^{uf(x)}-1\\
&=p^{s-1}-1.
\end{align*}
Thus, the desired conclusion is obtained.
\end{proof}


\begin{lemma}\label{lem:wt}
Let $\alpha\in \mathbb{F}_{q_2}^{\ast}$ and $D$ be defined in \eqref{set:D1} and $C_{D}$ be defined in \eqref{defcode1}. Let $\c_{(u,v)}$ be the corresponding codeword in $C_{D}$ with non-zero element $(u,v)\in \F$. We have the following.
\begin{enumerate}
\item [(1)] When $v\in \mathbb{F}_{q_2}\setminus \mathbb{F}_{p}^*\alpha$, we have $\wt(c_{(u,v)})=p^{s-2}(p-1)$.
\item [(2)] When $v\in \mathbb{F}_{p}^{\ast}\alpha$, we have the following four cases.
\begin{enumerate}
  \item [(2.1)] If $u\notin \textrm{Im}(L_f)$, then $\wt(c_{(u,v)})=p^{s-2}(p-1)$.

  \item [(2.2)] If $u\in \textrm{Im}(L_f)$, then
$$
\wt(c_{(u,v)})
=\left\{\begin{array}{ll}
p^{s-2}\Big(p-1-\varepsilon_{f}\bar{\eta}(f(x_{u}))(p^{\ast})^{-\frac{R_f-1}{2}}\Big), & \textrm{if\ } 2\nmid R_f , \\
p^{s-2}\Big(p-1-v(f(x_{u}))\varepsilon_{f}(p^{\ast})^{-\frac{R_f}{2}}\Big),  & \textrm{if\ } 2|R_f  .
\end{array}
\right.
$$
\end{enumerate}
\end{enumerate}
\end{lemma}
\begin{proof}
Put
$N(u,v)=\Big\{(x,y)\in \F: f(x)+\Tr_1^{s_2}(\alpha y)=0, \Tr((u,v)\odot (x,y)) = 0\Big\}$, then the Hamming weight of $\c_{(u,v)}$
is $n-|N(u,v)|+1$, where $n_1$ is given in Lemma \ref{lem:length}. Thus, we just need to evaluate the value of $|N(u,v)|$.

By the orthogonal property of additive characters, we have
\begin{align}\label{eq:N(u,v)}
&|N(u,v)|=\frac{1}{p^{2}}\sum_{(x,y)\in \mathbb{F}}\Big(\sum_{z_{1}\in \mathbb{F}_{p}}\zeta_{p}^{z_{1}f(x)+\Tr_1^{s_2}(z_{1}\alpha y)}\sum_{z_{2}\in \mathbb{F}_{p}}\zeta_{p}^{z_{2}\mathrm{Tr}((u,v)\odot (x,y))}\Big)  \\ \nonumber
&=\frac{1}{p^{2}}\sum_{(x,y)\in \mathbb{F}}\Big(\big(1+\sum_{z_{1}\in \mathbb{F}_{p}^{\ast}}\zeta_{p}^{z_{1}f(x)+\Tr_1^{s_2}(z_{1}\alpha y)}\big)\big(1+\sum_{z_{2}\in\mathbb{F}_{p}^{\ast}}\zeta_{p}^{z_{2}\mathrm{Tr}((u,v)\odot (x,y))}\big)\Big)  \\ \nonumber
&=p^{s-2}+\frac{1}{p^{2}}\sum_{z_{1}\in \mathbb{F}_{p}^{\ast}}\sum_{(x,y)\in \mathbb{F}}\zeta_{p}^{z_{1}f(x)+\Tr_1^{s_2}(z_{1}\alpha y)}+ \frac{1}{p^{2}}\sum_{z_{2}\in \mathbb{F}_{p}^{\ast}}\sum_{(x,y)\in \mathbb{F}}\zeta_{p}^{z_{2}\mathrm{Tr}((u,v)\odot (x,y))} \\ \nonumber
&+\frac{1}{p^{2}}\sum_{z_{1}\in \mathbb{F}_{p}^{\ast}}\sum_{z_{2}\in \mathbb{F}_{p}^{\ast}}\sum_{(x,y)\in \mathbb{F}}\zeta_{p}^{z_{1}f(x)+\Tr_1^{s_1}(z_{2}ux)+\Tr_1^{s_2}(z_{2}vy+z_{1}\alpha y)}  \\  \nonumber
&=p^{s-2}+p^{-2}\sum_{z_{1}\in \mathbb{F}_{p}^{\ast}}\sum_{z_{2}\in \mathbb{F}_{p}^{\ast}}\sum_{y\in \mathbb{F}_{q_1}}\zeta_{p}^{\Tr_1^{s_2}(z_{2}vy+z_{1}\alpha y)}\sum_{x\in \mathbb{F}_{q_1}}\zeta_{p}^{z_{1}f(x)+\Tr_1^{s_1}(z_{2}ux)}.
\end{align}

(1) When $v\in \mathbb{F}_{q_2}\setminus \mathbb{F}_{p}^{\ast}\alpha$, the desired conclusion then follows from $\sum_{y\in \mathbb{F}_{q_2}}\zeta_{p}^{\Tr_1^{s_2}(z_{2}vy+z_{1}\alpha y)}=0$.

(2) When $v\in \mathbb{F}_{p}^{\ast}\alpha$, i.e., $\alpha=zv$ for some $z\in \mathbb{F}_{p}^{\ast}$, \eqref{eq:N(u,v)} becomes
\begin{align}\label{eq:N(u,v):2}
|N(u,v)|
&=p^{s-2}+p^{s_2-2}\sum_{z_{1}\in \mathbb{F}_{p}^{\ast}}\sum_{x\in \mathbb{F}_{q_1}}\zeta_{p}^{z_{1}f(x)-z_{1}\Tr_1^{s_1}(zux)}\nonumber\\
&=p^{s-2}+p^{s_2-2}\sum_{z_{1}\in \mathbb{F}_{p}^{\ast}}\sigma_{z_{1}}\Big(\sum_{x\in \mathbb{F}_{q_1}}\zeta_{p}^{f(x)-\Tr_1^{s_1}(zux)}\Big).
\end{align}

(2.1) If $u\notin \textrm{Im}(L_f)$, by Lemma~\ref{lem:6}, we have $|N(u,v)|=p^{s-2}$, the desired conclusion of (2.1) is obtained.

(2.2) If $u\in \textrm{Im}(L_f)$,
define $c=zu$, we have $x_{c}=zx_{u}$. By Lemma~\ref{lem:6}, \eqref{eq:N(u,v):2} becomes
\begin{align}
|N(u,v)|
&=p^{s-2}+p^{s_2-2}\sum_{z_{1}\in \mathbb{F}_{p}^{\ast}}\sigma_{z_{1}}\Big(\varepsilon_{f}(p^{\ast})^{\frac{R_{f}}{2}}p^{s_1-R_f}\zeta_{p}^{-f(x_{c})}\Big) \nonumber \\
&=p^{s-2}+p^{s-2-R_f}\varepsilon_{f}\sum_{z_{1}\in \mathbb{F}_{p}^{\ast}}\sigma_{z_{1}}\Big((p^{\ast})^{\frac{R_{f}}{2}}\zeta_{p}^{-f(x_{c})}\Big). \nonumber
\end{align}
The conclusion  of (2.2) follows directly from Lemma~\ref{lem:7}.

\end{proof}



\begin{theorem}\label{thm:wd} Let $\alpha \in \F_{q_2}^*$ and $f$ be a homogeneous quadratic function defined in \eqref{eq:f}. Let $D$ be defined in \eqref{set:D1} and the code $C_{D}$ be defined in \eqref{defcode1}. Write $s=s_1+s_2$.
Then the code $C_{D}$ is a $[p^{s-1}-1,s]$ linear code over $ \mathbb{F}_{p} $
with the weight distribution in Tables 1 and 2.
\begin{table}\label{tab:wd:o}
\centering
\caption{The weight distribution of $C_{D}$ of Theorem \ref{thm:wd} when $R_f$ is odd}
\begin{tabular*}{10.5cm}{@{\extracolsep{\fill}}ll}
\hline
\textrm{Weight} $\omega$ & \textrm{Multiplicity} $A_\omega$   \\
\hline
0 &   1  \\
$(p-1)p^{s-2}$ &  $p^{s}-p^{R_f-1}(p-1)^2-1$  \\
$p^{s-2}\Big(p-1-p^{-\frac{R_f-1}{2}}\Big)$  & $\frac{1}{2}(p-1)^2p^{R_f-1}(1+p^{-\frac{R_f-1}{2}})$  \\
$p^{s-2}\Big(p-1+p^{-\frac{R_f-1}{2}}\Big)$  & $\frac{1}{2}(p-1)^2p^{R_f-1}(1-p^{-\frac{R_f-1}{2}})$  \\
\hline
\end{tabular*}
\end{table}
\begin{table}
\centering
\caption{The weight distribution of $C_{D}$ of Theorem \ref{thm:wd} when $R_f$ is even}
\begin{tabular*}{10.5cm}{@{\extracolsep{\fill}}ll}
\hline
\textrm{Weight} $\omega$ \qquad& \textrm{Multiplicity} $A_\omega$   \\
\hline
0 \qquad&   1  \\
$p^{s-2}(p-1)$ \qquad&  $p^{s}-p^{R_f}(p-1)-1$  \\
$p^{s-2}(p-1)\Big(1-\varepsilon_{f}(p^*)^{-\frac{R_f}{2}}\Big)$  \qquad& $(p-1)p^{R_f-1}\Big(1+\varepsilon_{f}(p-1)(p^*)^{-\frac{R_f}{2}}\Big)$  \\
$p^{s-2}(p-1+\varepsilon_{f}(p^*)^{-\frac{R_f}{2}})$  \qquad& $(p-1)^{2}p^{R_f-1}\Big(1-\varepsilon_{f}(p^*)^{-\frac{R_f}{2}}\Big)$  \\
\hline
\end{tabular*}
\end{table}
\end{theorem}

\begin{proof}
By Lemma \ref{lem:wt}, we know that, for $(u,v)\in\F^\star$, we have $\wt(\c_{(u,v)})>0$. So, the map: $\F\rightarrow C_{D}$ defined by $(u,v)\mapsto \c_{(u,v)} $
is an isomorphism in linear spaces over $\F_p$. Hence, the dimension of the code $C_{D}$ in \eqref{defcode1} is equal to $s$. By Lemma \ref{lem:length}, we proved that the code $C_{D}$ is a $[p^{s-1}-1,s]$ linear code over $ \mathbb{F}_{p} $.

Now we shall prove the multiplicities $A_{\omega_{i}}$ of codewords with weight $\omega_{i}$ in $C_{D}$. Let us give the proofs of two cases, respectively.

(1)  The case that $R_f$ is odd.

For each $(u,v)\in \F$ and $ (u,v) \neq (0,0)$. By Lemma \ref{lem:length} and Lemma \ref{lem:wt},
$\wt(c_{(u,v)})$ has only three values, that is,
$$
\left\{\begin{array}{ll}
\omega_{1}=(p-1)p^{s-2}, \  \\
\omega_{2}=p^{s-2}\Big(p-1-p^{-\frac{R_f-1}{2}}\Big), \  \\
\omega_{3}=p^{s-2}\Big(p-1+p^{-\frac{R_f-1}{2}}\Big).
\end{array}
\right.
$$

By Lemma \ref{lem:wt} and Lemma \ref{lem:1}, we have
\begin{align*}\label{eq:1}
  &A_{\omega_{1}} = \Big|\Big\{(u,v)\in\F|\wt(c_{(u,v)}) = (p-1)p^{s-2}\Big\}\Big|  \\
  &= \Big|\Big\{(u,v)\in\F^\star|u\in\F_{q_1},v\in\F_{q_2}\setminus\F_p^*\alpha\Big\}\Big|+\Big|\Big\{(u,v)\in\F|\\
  &u\notin\textrm{Im}(L_f),v\in\F_p^*\alpha\Big\}\Big|+\Big|\Big\{(u,v)\in\F|v\in\F_p^*\alpha, f(x_u)=0\Big\}\Big|\\
&=\Big(q_1(q_2-(p-1))-1\Big)+(p-1)+(q_1-p^{R_f})(p-1)\\
&+(p-1)\Big|\{\bar{x}\in\F_{q_1}/\Ker L_f|f(x)=0\}\Big|\\
&=\Big(q_1(q_2-(p-1))-1\Big)+(q_1-p^{R_f})(p-1)+(p-1)\frac{p^{s_1-1}}{p^{s_1-R_f}}\\
&=p^{s}-p^{R_f-1}(p-1)^2-1,
\end{align*}
where we use the fact that the number of solutions of the equation $f(x)=0$ in $\F_q$ is $p^{s_1-1}$ and the dimension of the kernel $L_f$ is $s_1-R_f$.

Similarly, the values of $A_{\omega_{2}}$ and $A_{\omega_{3}}$ can be calculated. This completes the proof of the weight distribution of Table 1.

 (2) The case that $R_f$ is even.

 The proof is similar to that of Case (1) and we omit it here. The desired conclusion then follows from Lemma \ref{lem:length}, Lemma \ref{lem:wt} and Lemma \ref{lem:1}.

\end{proof}


\begin{example}
Let $(p,s_1,s_2,\alpha)=(3,4,1,1)$ and $f(x)=\mathrm{Tr}_1^{s_1}(x^{2})$, by Corollary 1 in \cite{19TXF17}, we have
$\varepsilon_{f}=-1$ and $R_f=4$. Then, the corresponding code $C_{D}$ has parameters $[ 80,5,51]$ and the weight enumerator
$1+120x^{51}+80x^{54}+42x^{60}$, which is verified by the Magma program.
\end{example}

\begin{example}
Let $(p,s_1,s_2,\alpha)=(3,4,1,1)$ and $f(x)=\mathrm{Tr}_1^{s_1}(\theta x^{2})$, where $\theta$ is a primitive element of $\mathbb{F}_{p^{s_1}}$. By Corollary 1 in \cite{19TXF17}, we have
$\varepsilon_{f}=1$ and $R_f=4$. Then, the corresponding code $C_{D}$ has parameters $[80,5,48]$ and the weight enumerator
$1+66x^{48}+80x^{54}+96x^{57}$, which is verified by the Magma program.
\end{example}

\begin{example}
Let $(p,s_1,s_2,\alpha)=(3,4,1,1)$ and $f(x)=\mathrm{Tr}_1^{s_1}( x^{2})-\frac{1}{4}(\mathrm{Tr}_1^{s_1}( x))^{2}$. By Corollary 2 in \cite{19TXF17}, we have
$\varepsilon_{f}=1$ and $R_f=3$. Then, the corresponding code $C_{D}$ has parameters $[80,5,45]$ and the weight enumerator
$1+24x^{45}+206x^{54}+12x^{63}$, which is verified by the Magma program.
\end{example}


\subsection{The weight hierarchy of the presented linear code}

In this subsection, we give the weight hierarchy of $C_{D}$ in \eqref{defcode1}.

By Theorem~\ref{thm:wd}, we know that the dimension of the code $C_{D}$ defined in \eqref{defcode1} is $s$.
So, by Proposition \ref{pro:d_r}, we give a general formula,
that is
\begin{align}
        d_{r}(C_{D})&=n-\max\Big\{|H_r^\perp\cap D|: H_r \in [\mathbb{F},r]_{p}\Big\}  \\
        &=n-\max\Big\{|H_{s-r}\cap \D|: H_{s-r} \in [\mathbb{F},s-r]_{p}\Big\}\label{eq:d_r:2},
\end{align}
here $H_r^\perp = \Big\{(x,y)\in\F:\Tr((x,y)\odot (u,v)) =0, \textrm{for any $(u,v)\in H_r$}\Big \}$.

Let $H_r$ be an $r$-dimensional subspace of $\mathbb{F}$, set $
N(H_r)=\Big\{(x,y)\in \mathbb{F}: f(x)+\mathrm{Tr}_1^{s_2}(\alpha y)=0, \mathrm{Tr}((x,y)\odot (u,v))=0, (u,v)\in H_r \Big\}$.
Then, $N(H_r) = (H_r^\perp\cap D)\cup\{(0,0)\}$, which concludes that $|N(H_r)|=|H_r^\perp\cap\D |+1$. Hence, we have
\begin{equation}\label{eq:d_r:3}
     d_{r}(C_{\D})=n+1-\max\Big\{N(H_r)|: H_r \in [\mathbb{F},r]_{p}\Big\}.
\end{equation}

\begin{lemma} \label{lem:d_r:2}
Let $\alpha \in\F_{q_2}^*$ and $f$ be a homogeneous quadratic function defined in \eqref{eq:f} with the sign $\varepsilon_f$ and the rank $R_f$. $H_r$ and $N(H_r)$ are defined as above. We have the following.
\begin{itemize}
  \item[(1)] If $\alpha\notin  \Prj_{2}(H_r)$, then $N(H_r)=p^{s-(r+1)}$.
  \item[(2)] If $\alpha\in  \Prj_{2}(H_r)$, then
  $$N(H_r)=p^{s-(r+1)}\Big(1+\varepsilon_{f}p^{-R_f}\sum\limits_{(u,-\alpha)\in H_r,u\in\Im L_f}\sum\limits_{z\in \F_p^{\ast}}\sigma_{z}\big((p^{\ast})^{\frac{R_f}{2}}\zeta_{p}^{-f(x_{u})}\big)\Big).$$
\end{itemize}
Here $ \Prj_{2}$ is the second projection from $\mathbb{F}$ to $\mathbb{F}_{q_2}$ defined by $(x,y)\mapsto y$.
\end{lemma}

\begin{proof}
By the orthogonal property of additive characters, we have
\begin{align*}
&p^{r+1}|N(H_r)|
=\sum_{\x=(x,y)\in \F}\sum_{z\in \F_p}\zeta_{p}^{z f(x)+\Tr_1^{s_2}(z(\alpha y))}\sum_{\y\in H_r}\zeta_{p}^{\Tr(\x\odot\y)} \\
&=\sum_{\x=(x,y)\in \F}\sum_{\y\in H_r}\zeta_{p}^{\Tr(\x\odot \y)}
+\sum_{\x=(x,y)\in \F}\sum_{z\in \F_p^{\ast}}\sum_{\y\in H_r}\zeta_{p}^{z f(x)+\Tr(\x\odot \y)+\Tr_1^{s_2}(z \alpha y)} \\
&=p^s+\sum\limits_{(u,v)\in H_r}\sum_{z\in \F_p^{\ast}}\sum_{x\in \F_{q_1}}\zeta_{p}^{z f(x)+z\Tr_1^{s_1}(ux)}\sum_{y\in \F_{q_2}}\zeta_{p}^{z\Tr_1^{s_2}(vy+\alpha y)}
\end{align*}

If $\alpha\notin  \Prj_{2}(H_r)$, then $N(H_r)=p^{s-r-1}$, which follows from $\sum_{y\in \mathbb{F}_{q_2}}\zeta_{p}^{z\Tr_1^{s_2}(vy+\alpha y)}=0$.

If $\alpha\in  \Prj_{2}(H_r)$, by Lemma~\ref{lem:6}, we have
\begin{align*}
p^{r+1}|N(H_r)|&=p^s+q_2\sum_{(u,-\alpha)\in H_r}\sum_{z\in \F_p^{\ast}}\sum_{x\in \F_{q_1}}\zeta_{p}^{z f(x)+z\Tr_1^{s_1}(ux)}    \\
&=p^s+q_2\sum_{(u,-\alpha)\in H_r}\sum_{z\in \F_p^{\ast}}\sigma_{z}\Big(\sum_{x\in \F_{q_1}}\zeta_{p}^{f(x)+\Tr_1^{s_1}(ux)}\Big)    \\
&=p^s+\varepsilon_{f}p^{s-R_f}\sum_{(u,-\alpha)\in H_r,u\in\Im L_f}\sum_{z\in \F_p^{\ast}}\sigma_{z}\Big((p^{\ast})^{\frac{R_f}{2}}\zeta_{p}^{-f(x_{u})}\Big).
\end{align*}

So, the desired result is obtained. Thus, we complete the proof.
\end{proof}

In the following, we shall determine the weight hierarchy of $C_{D}$ in \eqref{defcode1}  by calculating $N(H_r)$ in Lemma~\ref{lem:d_r:2} and $|H_{s-r}\cap D|$ in \eqref{eq:d_r:2}.

\begin{theorem}\label{thm:wh} Let $\alpha \in \F_{q_1}^*$ and $f$ be a homogeneous quadratic function defined in \eqref{eq:f} with the sign $\varepsilon_f$ and the rank $R_f\geq 3$. Let $D$ be defined in \eqref{set:D1} and the code $C_{D}$ be defined in \eqref{defcode1}. Write $s=s_1+s_2$ and define
$$
e_0=\left\{\begin{array}{ll}
s_1-\frac{R_f+1}{2}, & \textrm{if $R_f$ is odd}, \\
s_1-\frac{R_f}{2}, & \textrm{if $R_f$ is even and $\varepsilon_{f}=(-1)^{\frac{R_f(p-1)}{4}}$},\\
s_1-\frac{R_f+2}{2}, & \textrm{if $R_f$ is even and $\varepsilon_{f}=-(-1)^{\frac{R_f(p-1)}{4}}$}.
\end{array}
\right.
$$
Then we have the following.
\begin{itemize}
\item[(1)]
When $s_1-e_0+1 \leq r \leq s$, we have
$$
d_{r}(C_{\D})=p^{s-1}-p^{s-r}.
$$
\item[(2)] When $0< r\leq s_1-e_0$, we have
$$
d_{r}(C_{\D})=\left\{\begin{array}{ll}
p^{s-1}-p^{s-1-r}-p^{s-1-\frac{R_f+1}{2}},  \textrm{if $2\nmid R_f$\ }, \\
p^{s-1}-p^{s-1-r}-(p-1)p^{s-1-\frac{R_f+2}{2}},  \textrm{if $2\mid R_f$ and $\varepsilon_{f}=(-1)^{\frac{R_f(p-1)}{4}}$},\\
p^{s-1}-p^{s-1-r}-p^{s-1-\frac{R_f+2}{2}},  \textrm{if $2\mid R_f$ and $\varepsilon_{f}=-(-1)^{\frac{R_f(p-1)}{4}}$}.
\end{array}
\right.
$$
\end{itemize}
\end{theorem}

\begin{proof}
(1) When $s_1-e_0+1 \leq r \leq s$, then $0 \leq s-r \leq s_2+e_0-1$.
 Let $T_{\alpha}=\Big\{x\in \mathbb{F}_{q_1}:\ \mathrm{Tr}_1^{s_2}(\alpha x)=0\Big\}$.
It is easy to know that $\dim(T_{\alpha})=s_1-1$.
By Lemma~\ref{lem:2}, there exists an $e_0$-dimensional subspace $J_{e_0}$ of $\F_{q_1}$ such that $f(x)=0$ for any $x\in J_{e_0}$.
Note that the dimension of the subspace $J_{e_0}\times T_{\alpha}$ is $e_0+s_2-1$. Let $H_{s-r}$ be a $(s-r)$-dimensional subspace of $J_{e_0}\times T_{\alpha}$,
then, $|H_{s-r}\cap D|=p^{s-r}-1$.
Hence, by \eqref{eq:d_r:2}, we have
$$
d_{r}(C_{\D})=n-\max\Big\{|D \cap H|: H \in [\mathbb{F},s-r]_{p}\Big\}=p^{s-1}-p^{s-r}.
$$
Thus, it remains to determine $d_r(C_{\D})$ when $0 < r \leq s_1-e_0$.

(2) When $0 < r\leq s_1-e_0$, we discuss case by case.

\textbf{Case 1:}\ $R_f (R_f\geq 3)$ is odd. In this case, $e_0 = s_1-\frac{R_f+1}{2}$ and $s_1-e_0 = \frac{R_f+1}{2}$, that is, $0< r\leq \frac{R_f+1}{2}$.

Let $H_r$ be an $r$-dimensional subspace of $\F$. If $\alpha\in  \Prj_{2}(H_{r})$, by Lemma~\ref{lem:7} and Lemma~\ref{lem:d_r:2}, we have
\begin{align*}
N(H_r)=p^{s-(r+1)}\Big(1+(-1)^{\frac{(R_f-1)(p-1)}{4}}p^{-\frac{R_f-1}{2}}\varepsilon_{f}\sum_{(u,-\alpha)\in H_{r},y_1\in\Im L_f}\bar{\eta}(f(x_{u})\big)\Big).
\end{align*}

For $0 < r \leq \frac{R_f+1}{2}$, we want to construct $H_r$ such that $N(H_r)$ reaches its maximum, that is, the number of such as $(u,-\alpha)$ is maximal in $H_{r}$ and for any $(u,-\alpha)\in H_{r}, y_1\in\Im \ L_f$ and $ \bar{\eta}(f(x_{u})) = (-1)^{\frac{(R_f-1)(p-1)}{4}}\varepsilon_{f}$. The constructing method is as follows.

 Taking an element $a\in\F_p^*$ satisfying $\overline{\eta}(a)=(-1)^{\frac{(R_f-1)(p-1)}{4}}\varepsilon_{f}$, by Lemma \ref{lem:2-2},
there exists an $r$-dimensional subspace $J_{r}$ of $\F_{q_1}$ such that $R_{J_{r}}=1, \varepsilon_{J_r}=\overline{\eta}(a)$ and $J_r\cap \F_{q_1}^{\perp_f}=\{0\}$, which concludes that there exists an $(r-1)$-dimensional subspace $J_{r-1}$ of $J_{r}$ satisfying $f(J_{r-1})=0$ and
$\overline{\eta}(f(x))=\overline{\eta}(a)$, for each $x\in J_{r}\setminus J_{r-1}$.
Let $\alpha_{1},\alpha_{2},\cdots,\alpha_{r-1}$ be an $\mathbb{F}_{p}$-basis of $J_{r-1}$. Take an element $\alpha_{r}\in J_{r}\setminus J_{r-1}$ and
set
$$
\mu_{1}=\alpha_{1}+\alpha_{r},\mu_{2}=\alpha_{2}+\alpha_{r},\cdots,\mu_{r-1}=\alpha_{r-1}+\alpha_{r},\mu_{r}=\alpha_{r},
$$
clearly, $\mu_{1},\mu_{2},\cdots,\mu_{r-1},\mu_{r}$ is an $\mathbb{F}_{p}$-basis of $J_{r}$.

Take
$$
H_{r}=\Big\langle(L_f(\mu_{1}),-\alpha),(L_f(\mu_{2}),-\alpha),\cdots,(L_f(\mu_{r-1}),-\alpha),(L_f(\mu_{r}),-\alpha)\Big\rangle
,$$
it is easily seen that $H_{r}$ is our desired $r$-dimensional subspace of $\F$ and its $N(H_r)$ reaches the maximum
$$ N(H_r)=p^{s-(r+1)}\Big(1+p^{r-1-\frac{R_f-1}{2}}\Big) = p^{s-r-1} + p^{s-1-\frac{R_f+1}{2}}. $$
So, the desired result is obtained by Lemma~\ref{lem:d_r:2} and \eqref{eq:d_r:3}.

For $r = \frac{R_f+1}{2}$, let $H_r$ be an $r$-dimensional subspace of $\F$. By Lemma~\ref{lem:7} and Lemma~\ref{lem:d_r:2}, we have
$N(H_r)\leq 2p^{s-1-\frac{R_f+1}{2}}$,
which concludes that $$ d_{r}(C_{\D})\geq p^{s-1} - 2p^{s-1-\frac{R_f+1}{2}}$$
by formula \eqref{eq:d_r:3}. On the other hand, by formula \eqref{eq:d_r:2}, we have
$$ d_{r}(C_{\D}) = p^{s-1} -1 - \max\Big\{|H\cap D|: H \in [\mathbb{F},s-\frac{R_f+1}{2}]_{p}\Big\}.$$
Now we want to construct a $(s-\frac{R_f+1}{2})$-dimensional subspace $H$ of $\F$ such that $|H\cap D|\geq 2p^{s-1-\frac{R_f+1}{2}}-1$, which concludes that
$$ d_{r}(C_{\D})\leq p^{s-1} - 2p^{s-1-\frac{R_f+1}{2}}.$$
In fact, by the proof of (1), we know that the dimension of  $ J_{e_0}$ is $ s_1-\frac{R_f+1}{2}$.
Taking $(u,v)\in D$, where $u\in J_{e_0}^{\perp_f}$ and $f(u) \neq 0$, define $H = (J_{e_0}\times T_{\alpha}) \oplus \Big\langle(u,v) \Big\rangle$,
then $H$ is our desired $(s-\frac{R_f+1}{2})$-dimensional subspace of $\F$. So, for $r = \frac{R_f+1}{2}$, we have
$$d_{r}(C_{\D}) = p^{s-1} - 2p^{s-1-\frac{R_f+1}{2}}.$$
\textbf{Case 2:}\ $R_f(R_f\geq 3) $ is even and $\varepsilon_{f}=(-1)^{\frac{R_f(p-1)}{4}}$. In this case, $e_0 = \frac{R_f}{2}$ and $R_f-e_0 = \frac{R_f}{2}$, that is, $0\leq r\leq \frac{R_f}{2}$.

Suppose $H_r$ is an $r$-dimensional subspace of $\F$ and $\alpha\in  \Prj_{2}(H_{r})$. Recall that $v(0)=p-1$ and $v(x)=-1$ for $x\in\mathbb{F}_{p}^{\ast}$ defined in Subsection~\ref{subsec:notation}.
By Lemma~\ref{lem:7} and Lemma~\ref{lem:d_r:2}, we have
\begin{align*}
N(H_r)&=p^{s-(r+1)}\Big(1+\varepsilon_{f}p^{-R_f}\sum\limits_{(u,-\alpha)\in H_r,u\in\Im L_f}\sum\limits_{z\in \F_p^{\ast}}\sigma_{z}\big((p^{\ast})^{\frac{R_f}{2}}\zeta_{p}^{-f(x_{u})}\big)\Big)\\
&=p^{s-(r+1)}\Big(1+p^{-\frac{R_f}{2}}\sum\limits_{(u,-\alpha)\in H_r,u\in\Im L_f}v(f(x_{u})\Big).
\end{align*}

Let $J_{r}=\Big\langle\overline{\mu}_{1},\overline{\mu}_{2},\cdots,\overline{\mu}_{r}\Big\rangle$ be a $r$-dimensional subspace of $J_{e_0}\Big/\F_{q_1}^{\perp_f}$.
Take
$$
H_{r}=\Big\langle(L_f(\overline{\mu}_{1}),-\alpha),(L_f(\overline{\mu}_{2}),-\alpha),\cdots,(L_f(\overline{\mu}_{r}),-\alpha)\Big\rangle
,$$
then  $N(H_r)$ reaches its maximum
$$ N(H_r)=p^{s-(r+1)}\Big(1+(p-1)p^{r-1-\frac{R_f}{2}}\Big) = p^{s-1-r} + (p-1)p^{s-2-\frac{R_f}{2}}. $$
So, for $0\leq r\leq\frac{R_f}{2}$, the desired result is obtained by Lemma~\ref{lem:d_r:2} and \eqref{eq:d_r:3}.

\textbf{Case 3}:\ $R_f (R_f\geq 3)$ is even and $\varepsilon_{f}=-(-1)^{\frac{R_f(p-1)}{4}}$. In this case, $e_0 = \frac{R_f-2}{2}$ and $R_f-e_0 = \frac{R_f}{2}+1$, that is, $0\leq r\leq \frac{R_f}{2}+1$.

Suppose $H_r$ is an $r$-dimensional subspace of $\F$ and $\alpha\in  \Prj_{2}(H_{r})$.
By Lemma~\ref{lem:7} and Lemma~\ref{lem:d_r:2}, we have
\begin{align}
N(H_r)=p^{s-(r+1)}\Big(1-p^{-\frac{R_f}{2}}\sum\limits_{(u,-\alpha)\in H_{r},u\in\Im L_f}v(f(x_{u})\Big).   \nonumber
\end{align}

For $0\leq r\leq\frac{R_f}{2}$, taking an element $a\in\F_p^*$, by Lemma~\ref{lem:2-2}, there exists an $r$-dimensional subspace $J_{r}$ of $\F_{q_1}$ such that $R_{J_{r}}=1, \varepsilon_{J_r}=\overline{\eta}(a)$ and $J_r\cap \F_{q_1}^{\perp_f}=\{0\}$, which concludes that there exists an $(r-1)$-dimensional subspace $J_{r-1}$ of $J_{r}$ satisfying $f(J_{r-1})=0$ and
$f(x)=a$, for each $x\in J_{r}\setminus J_{r-1}$.
Let $\alpha_{1},\alpha_{2},\cdots,\alpha_{r-1}$ be an $\mathbb{F}_{p}$-basis of $J_{r-1}$. Take an element $\alpha_{r}\in J_{r}\setminus J_{r-1}$ and
set
$$
\mu_{1}=\alpha_{1}+\alpha_{r},\mu_{2}=\alpha_{2}+\alpha_{r},\cdots,\mu_{r-1}=\alpha_{r-1}+\alpha_{r},\mu_{r}=\alpha_{r},
$$
it's obvious that $\mu_{1},\mu_{2},\cdots,\mu_{r-1},\mu_{r}$ is an $\mathbb{F}_{p}$-basis of $J_{r}$.
Take
$
H_{r}=\Big\langle(L_f(\mu_{1}),-\alpha),(L_f(\mu_{2}),-\alpha),\cdots,(L_f(\mu_{r-1}),-\alpha),(L_f(\mu_{r}),-\alpha)\Big\rangle
$,
then $N(H_r)$ reaches its maximum
$$ N(H_r)=p^{s-(r+1)}\Big(1+p^{r-1-\frac{R_f}{2}}\Big) = p^{s-r-1} + p^{s-1-\frac{R_f+2}{2}}. $$
So, the desired result is obtained by Lemma~\ref{lem:d_r:2} and \eqref{eq:d_r:3}.

For $r = \frac{R_f}{2}+1$, let $H_r$ be an $r$-dimensional subspace of $\F$. By Lemma~\ref{lem:7} and Lemma~\ref{lem:d_r:2} and formula \eqref{eq:d_r:3}, we have
$N(H_r)\leq 2p^{s-1-\frac{R_f+2}{2}}$,
which concludes that $$ d_{r}(C_{D})\geq p^{s-1} - 2p^{s-1-\frac{R_f+2}{2}}.$$
On the other hand,
by the proof of (1), we know that the dimension of  $ J_{e_0}$ is $ s_1-\frac{R_f+2}{2}$.
Taking $(u,v)\in D$, where $u\in J_{e_0}^{\perp_f}$ and $f(u) \neq 0$, define $H = (J_{e_0}\times T_{\alpha}) \oplus \Big\langle(u,v) \Big\rangle$, $H$ is an $(s-1-\frac{R_f}{2})$-dimensional subspace of $\F$.
So by formula \eqref{eq:d_r:2}, we have
\begin{align*}
d_{r}(C_{\D}) &= p^{s-1} -1 - \max\Big\{|H\cap D|: H \in [\mathbb{F},s-1-\frac{R_f}{2}]_{p}\Big\}\\
&\leq p^{s-1}-2p^{s-1-\frac{R_f+2}{2}}.
\end{align*}

 Therefore, for $r = \frac{R_f}{2}+1$, we have
$$d_{r}(C_{\D}) = p^{s-1} - 2p^{s-1-\frac{R_f+2}{2}}.$$
The desired result is obtained.

\end{proof}

\begin{remark} We recall the definition of $r$-MDS code \cite{TV95}. Let $C$ be an $[n, k]$ linear code over $\F_p$. For any integer $1 \leq r \leq k$, it is
known that
$$r \leq d_r(C) \leq n - k + r,$$
and $C$ is called an $r$-MDS code if $d_r(C) = n - k + r$. It is easy to check that the codes $C_{D}$ defined in \eqref{defcode1} satisfy
that $d_s(C) = n - s + s$, that is, $C_{D}$ are $s$-MDS codes.
\end{remark}

\section{Concluding Remarks}

In the work of \cite{LL22}, using a special non-degenerate homogeneous quadratic function we constructed a family of three-weight linear codes and determined their weight distributions and weight hierarchies.
In this paper, our defining set is $ D=\Big\{(x,y)\in \Big(\F_{p^{s_1}}\times\F_{p^{s_2}}\Big)\Big\backslash\{(0,0)\}: f(x)+\Tr_1^{s_2}(\alpha y)=0\Big\}$,
where $\alpha\in\mathbb{F}_{p^{s_2}}^*$ and $f(x)$ is a quadratic forms over $\mathbb{F}_{q_1}$ with values in $\F_p$, whether $f(x)$ is non-degenerate or not. We also constructed a family of three-weight linear codes with any dimension by extending known construction and determined their weight distributions and weight hierarchies.

 Let $w_{\min}$ and $w_{\max}$ denote the minimum and maximum nonzero weight of our obtained code $C_{D}$ defined in \eqref{defcode1}, respectively.
If $R_f\geq3$, then it can be easily checked that
$$
 \frac{w_{\min}}{w_{\max}}> \frac{p-1}{p}.
$$
By the results in \cite{AB98} and \cite{21YD06}, we know that every nonzero codeword of $C_{D}$ is minimal and most of the codes we constructed
are suitable for constructing secret sharing schemes with interesting properties.

\section*{Acknowledgement}

For the research, the third author was supported by the National Science Foundation of China Grant No.12001312
and the second author was supported by Key Projects in Natural Science Research of Anhui Provincial Department of Education No.2022AH050594 and Anhui Provincial Natural Science Foundation No.1908085MA02.

\section*{Declaration of competing interest}

The authors declare that they have no known competing financial interests or personal relationships that could have appeared to influence the work
reported in this paper.

\end{document}